\documentclass[12pt,reqno]{amsart}
\usepackage{amssymb,mathrsfs}
\usepackage{amsaddr}
\usepackage[numbers,sort&compress]{natbib}
\usepackage{graphicx,epsfig}
\usepackage{enumitem}
\usepackage{caption}

\newcommand*{\at}{@}

\newtheorem*{thm}{Theorem}
\newtheorem{prop}{Proposition}
\newtheorem{lem}{Lemma}

\theoremstyle{remark}
\newtheorem{rem}{\bf Remark}
\newtheorem*{cor}{\bf Corollary}

\def\bbR{\mathbb{R}}

\def\pa{\partial}
\newcommand{\papa}[2]{\frac{\partial#1}{\partial#2}}

\newcommand{\nn}{\nonumber}

\def\dg{\dagger}
\def\df{\overset{\rm df}{=}}
\newcommand{\Tr}[1]{\mathop{{\mathrm{Tr}}_{#1}}}
\newcommand{\id}{\mathop{{\mathrm{id}}}\nolimits}
\newcommand{\rank}{\mathop{{\mathrm{rank}}}\nolimits}
\newcommand{\ave}[1]{\langle#1\rangle}

\def\GG{\mathfrak{G}}

\def\a{\alpha}

\def\d{\delta}

\def\e{\epsilon}

\def\vr{\varrho}

\def\s{\sigma}

\def\O{\mathcal{O}}
\def\D{\mathcal{D}}
\def\B{\mathcal{B}}

\def\C{\mathcal{C}}
\def\G{\mathcal{G}}
\def\I{\mathcal{I}}

\newcommand*{\GtrSim}{\smallrel\gtrsim}

\makeatletter
\newcommand*{\smallrel}[2][.85]{%
  \mathrel{\mathpalette{\smallrel@{#1}}{#2}}%
}
\newcommand*{\smallrel@}[3]{%
  \sbox0{$#2\vcenter{}$}%
  \dimen@=\ht0 %
  \raise\dimen@\hbox{%
    \scalebox{#1}{%
      \raise-\dimen@\hbox{$#2#3\m@th$}%
    }%
  }%
}
\makeatother

\begin{document}

\title[Coherent information of one-mode Gaussian channels]{Coherent information of one-mode Gaussian channels -- the general case of nonzero added classical noise}

\begin{abstract}
    We prove that whenever the coherent information of a one-mode Gaussian channel is non-zero its supremum is achieved for the infinite input power. This is a well established fact for the zero added classical noise, whereas the nonzero case has not  been studied in detail. The presented analysis fills the gap for three canonical classes of one-mode Gaussian channels: the lossy, amplifying and  additive classical noise channel class. For the remaining one-mode Gaussian channel classes the coherent information is known to vanish.
\end{abstract}


\author{Kamil Br\'adler}
\email{kbradler\at ap.smu.ca}
\address{
    Department of Astronomy and Physics,
    Saint Mary's University,
    Halifax, Nova Scotia, B3H 3C3, Canada
    }

\maketitle

\section{Introduction}

Quantum capacity of a noisy quantum channel $\G$ is the highest rate at which quantum information can be sent down the channel and faithfully recovered from the channel's output. The quantum capacity is given by the formula~\cite{devetak2005capacity,shor2002quantum,lloyd1997capacity}
\begin{equation}\label{eq:QuantCap}
Q(\G)=\lim_{n\to\infty}{1\over n}Q^{(1)}(\G^{\otimes n}),
\end{equation}
where $I_{\rm coh}(\G)\equiv Q^{(1)}(\G)=\sup_{\vr}I(A\rangle B)_\s$ is the one-shot coherent information~\cite{barnum1998information}. The supremum is taken over all input states $\vr$ to the channel $\G:\vr_A\mapsto\s_B$ and $I(A\rangle B)_\s\df H(B)_\s-H(BR)_\s$ is the \emph{non-optimized coherent information}~\cite{holevo2012quantum} where $\Tr{R}\s_{BR}=\s_B$. The von Neumann entropy $H(B)_\s$ is evaluated on the $B$ subsystem of $\s_{BR}$ and  $R$ denotes a purifying system of the input state $\vr_A$.  This is the best characterization of quantum communication capabilities of a quantum channel as of today and the situation is not entirely satisfactory. The need for regularization in Eq.~(\ref{eq:QuantCap}) has its origin in the fact that $Q^{(1)}(\G^{\otimes 2})\neq2 Q^{(1)}(\G)$ -- the coherent information is known to be non-additive~\cite{divincenzo1998quantum} for a general quantum channel $\G$. This ultimately means that $Q(\G)\neq Q^{(1)}(\G)$. Because of these problems it may happen that  a better characterization of quantum communication capabilities of a quantum channel will be developed which is not based on the coherent information. Until then, the coherent information is a useful lower bound but there are situations where even to calculate the coherent information can  be a challenge. One of them is the case of bosonic Gaussian channels. Bosonic Gaussian channels transform Gaussian optical states (such as a single-mode squeezed state) into Gaussian states and the name comes from a Gaussian profile of a probability function in phase space. Gaussian states and maps essentially describe the physics of a quantum harmonic oscillator as an example of an elementary infinite-dimensional quantum system. There are many reasons to study Gaussian states and channels. The whole Gaussian framework is simple enough to be  accessible to an analytical treatment and Gaussian states are relatively easily prepared and manipulated in a laboratory. At the same time, Gaussian channels show a number of complex phenomena such as superadditivity of the quantum capacity~\cite{smith2011quantum} and other effects~\cite{weedbrook2012gaussian}.

To make matters  more simple it is instructive to start with the one-mode Gaussian transformations. These maps correspond to the action of one-mode Gaussian (OMG) channels on (one-mode) Gaussian states whose complete classification has been  accomplished~\cite{holevo2007one}. Probably the second main reason behind the classification of OMG channels (the first one being the classification itself) is the ability to quantify the channels' quantum communication capabilities in terms of  channel capacities. Recall that they represent natural quantum optical processes such as the phase-insensitive amplifier or unbalanced beam-splitter~\cite{weedbrook2012gaussian}. If nothing else, the knowledge of how reliably they transmit quantum information can be of a great practical use. The question of the classical capacity of a OMG channel has been elucidated~\cite{giovannetti2013ultimate} (at least in the most important case of phase-insensitive OMG channels, see also~\cite{schafer2013equivalence,pilyavets2012methods}) and it has been shown that under reasonable conditions, the Holevo quantity is an additive quantity  for a broad class of OMG channels while using Gaussian codes~\cite{gioCMP}.

For the quantum capacity the situation is less encouraging~\cite{holevo2001evaluating,cerf2007quantum,wolf2007quantum,holevo2012quantum}. First of all, the quantum capacity is known only in two cases: (i) for so-called antidegradable channels~\cite{caruso2006one} where the capacity is zero and (ii) for a tiny fraction of OMG channels known to be degradable~\cite{devetakshor2005capacity}, where the capacity is nonzero and calculable~\cite{wolf2007quantum,caruso2006degradability,holevo2012quantum}. For the rest of OMG channels the generic computable lower bound is the one-shot coherent information as previously mentioned, where Gaussian states encode quantum information. But here is the catch -- unlike the classical capacity case, much less is known about whether the Gaussian codes provide the highest transmission rate~\cite{wolf2007quantum}. By restricting to the Gaussian codes, as we do in this work, the calculation of the one-shot coherent information consists of an optimization step over the unconstrained input mean photon number. As it turns out, only in a certain (very special) case the  coherent information is manifestly maximized when the input signal power goes to infinity (leading to a finite value of the coherent information)~\cite{holevo2012quantum}. In more detail, every OMG channel can be purified by a two-mode Gaussian unitary whose reference (environment) mode need not be  a pure state but rather  a mixed Gaussian state. This will have an effect on what OMG channel can be obtained by tracing over the reference mode. Ignoring it would result in an incomplete classification and missing important examples of OMG channels. So this is the key point. The coherent information can be relatively easily optimized if the environment contains no added classical noise. But the task becomes considerably harder in a general case.

Here we resolve this issue and show that if the environment contains added classical noise then whenever the coherent information is positive it is maximized for the infinite input signal power. Otherwise it equals zero as the only non-negative value for the coherent information and it happens if the input power is zero. The proof is greatly facilitated by an interesting identity (see Eq.~(\ref{eq:RHSderwrtKbounded}) for the lossy case and Eq.~(\ref{eq:interestingIdentityAmp}) for the amplifying case)  involving expressions derived from the von Neumann entropy for bosonic Gaussian systems  (see Eqs.~(\ref{eq:Fh}), (\ref{eq:fourPapas}) and (\ref{eq:fourPapasAmp}). The identity could be useful in solving other problems involving the capacities of Gaussian channels.

Note that from the text below Eq.~(5.8) in the seminal Ref.~\cite{holevo2001evaluating} it seems that the problem is  trivial. The authors claim that the coherent information is always an increasing function of the input power including arbitrary additive classical noise. However, a quick glance at Fig.~\ref{fig:Gexample} shows that this is not true. A more correct statement can be found in~\cite{holevo2012quantum}, where we can read below Eq.~(12.188) that the very same coherent information is a complicated function of the input power (see also the last paragraph of~\cite{holevo2007one}). We will show in this paper  that the in spite of the non-monotonic behavior of the coherent information as a function of the input power $N$, its supremum over all OMG input states is achieved in the limit of an infinite input power (such as the grey curve in Fig.~\ref{fig:Gexample}) or the optimized coherent information is zero (see the green and blue examples for $N=0$). This finding thus justifies the quantum capacity formula (5.9) from~\cite{holevo2001evaluating}. It is perhaps an expected result but the author is not aware of its  proof in any form.

\begin{figure}[t]
  \resizebox{10cm}{!}{\includegraphics{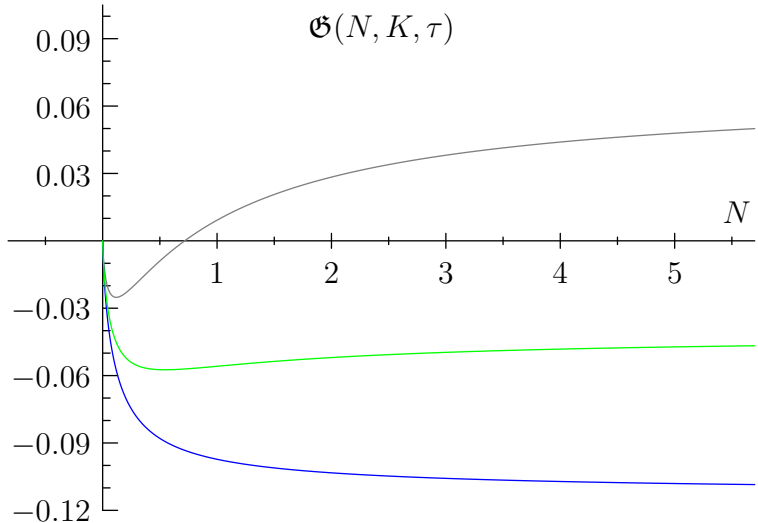}}
   \caption{The non-optimized coherent information defined in Eq.~(\ref{eq:Icoh}) is plotted to illustrate its non-trivial dependence on the input signal power $N$. We depict a lossy Gaussian channel $\G$ from the class $\C(\mathrm{loss})$ with the transmissivity parameter $\tau=2/3$ and the added classical noise parameter $K\geq0$ ($K=1/12$ for the grey curve, $K=7/65$ for the green curve and $K=1/8$ for the blue curve). For their definition, see Section~\ref{sec:main}.}
    \label{fig:Gexample}
\end{figure}

The question of optimization of coherent  information arose while working on the problem of quantum information  encoded in Gaussian states and its flow in the vicinity of a non-rotating black hole surrounded by a potential barrier~\cite{bradler2014black}. There, we have shown that for a certain region of the black hole parameter space, where the quantum capacity cannot be calculated, the coherent information  is zero in the limit of an infinite input power. This may naively imply that an observer at future infinity will not be able to reconstruct the quantum information swallowed by such a black hole. This would be an unwelcome conclusion as there exists a common belief (shared by the authors of~\cite{bradler2014black}) that the information is not lost after the black hole has evaporated. But the whole problem may become a non-issue by either showing that (i) the region in question is physically irrelevant or (ii) that the coherent information could be superadditive leading to nonzero capacity or (iii) the one-shot quantum capacity (the coherent information) reaches its (nonzero) maximum for a finite input signal power. As a result of this paper, the third option turn out not to be the case.

\section{OMG channels and their coherent information}
\label{sec:main}

Gaussian completely positive maps (Gaussian channels) transform  Gaussian states among themselves.  Gaussian $n$-mode states are fully characterized  by a $2n$-dimensional displacement vector $\mathbf{d}$ and a $2n\times2n$ real covariance matrix $\mathbf{V}$. Hence, Gaussian channels act on these two quantities and it turns out that the transformation is $\mathbf{d}\mapsto\mathbf{Td}+\boldsymbol{\d}$ and $\mathbf{V}\mapsto \mathbf{T}\mathbf{V}\mathbf{T}^T+\mathbf{N}$ subject to a complete positivity condition on real matrices $\mathbf{T}$ and $\mathbf{N}=\mathbf{N}^T$~\cite{holevo2001evaluating}. For the simplest class of one-mode Gaussian channels investigated here the complete positivity condition reduces to
\begin{equation}\label{eq:CPcondition}
  y\geq|\tau-1|,\hspace{5mm}\mathbf{N}\geq0,
\end{equation}
where $y\df\sqrt{\det{\mathbf{N}}}$ and $\tau\df\det{\mathbf{T}}$. OMG channels have been shown to fall into several equivalence classes~\cite{holevo2007one}, where each class is represented by $\boldsymbol{\d}$ (that can be set to zero for our purposes) and matrices $\mathbf{T}$ and $\mathbf{N}$ being diagonal and characterized by three parameters. Adding to the already introduced parameter $\tau$, the remaining parameters are $r\df\min{[\rank{\mathbf{T}},\rank{\mathbf{N}}]}$ and $\ave{n}$ as an average number of thermal photons injected to the reference mode of the purifying symplectic  transformation of an OMG channel.

The OMG channels studied here are all $r=2$ channels sometimes called  phase-insensitive OMG channels. It is the lossy channel class $\C(\mathrm{loss})$ defined for $0<\tau<1$, amplifying channel class $\C(\mathrm{amp})$ defined for $\tau>1$ and the so-called additive classical noise channel class $\B_2$ as a limiting case $\tau\to1$ of the previous two classes. Their equivalence representatives are
\begin{subequations}\label{eq:eqReps}
  \begin{align}
    (\mathbf{T},\mathbf{N})_{\C({\mathrm{loss}})} & = \big(\sqrt{\tau}\id,\ (1-\tau)(2\ave{n}+1)\id\big),\label{eq:canForms1a}\\
    (\mathbf{T},\mathbf{N})_{\C({\mathrm{amp}})} & = \big(\sqrt{\tau}\id,\ (\tau-1)(2\ave{n}+1)\id\big),\\
    (\mathbf{T},\mathbf{N})_{\B_2} & = \big(\id,\ave{n}\id\big),
  \end{align}
\end{subequations}
where $\id$ stands for an identity matrix. The fourth class of rank two OMG channels is conjugate amplifying channel class $\D$ whose quantum capacity (and so coherent information) is zero. Therefore, this class is useless for reliable transmission of quantum information~\cite{holevo2001evaluating,holevo2007one,caruso2006one}.
For comprehensive reviews of OMG channels we refer the reader to~\cite{holevo2012quantum,weedbrook2012gaussian}. Sch\"afer's thesis~\cite{schafer2013} possibly contains the most detailed account of OMG channels focused on the classical capacity problem.

Let's review the relevant entropic quantities first introduced in~\cite{holevo2001evaluating} with a slightly modified nomenclature. We first recall the expression for the optimized coherent information of a Gaussian channel $\G$ and introduce the \emph{non-optimized coherent information} $\GG(N,K,\tau)$:
\begin{equation}\label{eq:Icoh}
  I_{\rm coh}(\G)\df\sup_{N}\GG(N,K,\tau)=\sup_{N}\big[g(\eta)-g(f)-g(\ell)\big],
\end{equation}
where $N=\Tr{}[\vr\;a^\dg a]$ is the mean photon number of an input Gaussian state~$\vr$, $K\df|\tau-1|\ave{n}$ (which is valid for $\C(\mathrm{amp})$ and $\C(\mathrm{loss})$) leading to
$$
K=1/2(y-|\tau-1|),
$$
as follows from the definition of $y$ and Eqs.~(\ref{eq:eqReps}). This is the parameter representing the classical noise  added to the channel's environment. The previously introduced parameter $\tau$ modulates the loss or gain of the corresponding optical element (an unbalanced beam-splitter for $\C(\mathrm{loss})$ or a parametric amplifier for $\C(\mathrm{amp})$) and $g(x)$ is the von Neumann entropy~\cite{agarwal}
\begin{equation}\label{eq:g}
g(x)\df(1+x)\log{[1+x]}-x\log{x}\;.
\end{equation}
Modifying the notation from~\cite{holevo2001evaluating} (see also~\cite{holevo2012quantum}), we then
have
\begin{subequations}\label{eq:cohInfoPars}
\begin{align}\label{eq:cohInfoParsA}
  \eta & = \begin{cases}
        \tau N+K \quad&{\rm for}\ 0<\tau<1\\
        \tau N+\tau-1+K \quad&{\rm for}\ \tau>1,
  \end{cases}\\
  f & = {1\over2}(p+\eta-N-1),\label{eq:cohInfoParsB}\\
  \ell & = {1\over2}(p-\eta+N-1).\label{eq:cohInfoParsC}
\end{align}
\end{subequations}
We must not forget to introduce
\begin{equation}\label{eq:pakaD}
 p=\big[(N+\eta+1)^2-\tau\,4N(N+1)\big]^{1/2}
\end{equation}
and for future reference we also define
\begin{equation}
q=\begin{cases}\label{eq:q}
         K+N+2KN-N\tau &{\rm for}\ 0<\tau<1\\
         K+2KN+N(-1+\tau)+\tau  &{\rm for}\ \tau>1.
  \end{cases}
\end{equation}

\begin{thm}
    Let $\GG(N,K,\tau)=g(\eta)-g(\ell)-g(f)$ be the non-optimized coherent information of a Gaussian channel $\G\in\{\C(\mathrm{loss}),\C(\mathrm{amp})\}$. Then either
    \begin{equation}\label{eq:main}
        \sup_{N}\GG(N,K,\tau)=\lim_{N\to\infty}\GG(N,K,\tau)>0
    \end{equation}
    or
    \begin{equation}\label{eq:mainTriv}
        \sup_{N}\GG(N,K,\tau)=\lim_{N\to0}\GG(N,K,\tau)\equiv0
    \end{equation}  
    holds for all $K\geq0$ and $1/2<\tau<1$ (for $\G\in\C(\mathrm{loss})$) or $\tau>1$ (for $\G\in\C(\mathrm{amp})$).
\end{thm}
We present a few simple auxiliary lemmas in order to prove the theorem. We start with the lossy channel class $\C(\mathrm{loss})$. 

Note that the following result already appeared in~\cite{holevo2001evaluating} and we present it here for the sake of completeness.
\begin{lem}\label{lem:limitsNinfty}
    Assume $0<\tau<1$ and $K\geq0$. Then
  \begin{enumerate}[label=(\roman*)]
    \item
    \begin{equation}\label{eq:Ninfty1}
          \lim_{N\to\infty}g(f)=-{K\over1-\tau}\log{K\over1-\tau}+{1-\tau+K\over1-\tau}\log{1-\tau+K\over1-\tau},
    \end{equation}
      \item
    \begin{equation}\label{eq:Ninfty2}
         \lim_{N\to\infty}[g(\eta)-g(\ell)]=\log{\tau\over1-\tau}.
    \end{equation}
  \end{enumerate}
\end{lem}
\begin{proof}
  (i) Just for comfort we set $N\to1/N$ and expand $f$ around $N=0$ yielding $f={K\over1-\tau}+\O(N)$. By using
  \begin{equation}\label{eq:logExpansion}
    \log{\bigg[\sum_{k=0}^\infty a_k N^k\bigg]}=a_0+{a_1\over a_0}N+\O(N^2)
  \end{equation}
  and Eq.~(\ref{eq:g}) we immediately obtain~(\ref{eq:Ninfty1}) (for $N\to0$).\\
  (ii) In this limit both summands diverge for $N\to0$  (recall our transformation $N\to1/N$) but the infinities conveniently cancel. First, in a manner similar to the previous calculation, we obtain
  \begin{equation}\label{eq:Nprime}
    \lim_{N\to\infty}g(\eta)=\lim_{N\to0}\Big[1+\log{\tau}-\log{N}+\O(N)\Big].
  \end{equation}
  Before taking the limit we expand $\ell$ around zero:
  $$
  \ell={1-\tau\over N}+{K\tau\over1-\tau}+\O(N)
  $$
  and by using
  \begin{equation}\label{eq:logExpansion1}
    \log{\bigg[\sum_{k=-1}^\infty a_k N^k\bigg]}=\log{a_{-1}}-\log{N}+{a_0\over a_{-1}}N+\O(N^2)
  \end{equation}
  we find
  \begin{equation}\label{eq:ExpOfLog1}
    \log{\ell}=\log{[1-\tau]}-\log{N}+{NK\tau\over(1-\tau)^2}+\O(N^2).
  \end{equation}
  After the similar procedure for $1+\ell$ we find
  \begin{equation}
    \lim_{N\to\infty}g(\ell)  =  \lim_{N\to0}[1-\log{N}+\log{[1-\tau]}+\O(N)].
  \end{equation}
  By subtracting it from Eq.~(\ref{eq:Nprime}) we can finally take the limit and obtain the RHS of Eq.~(\ref{eq:Ninfty2}).
\end{proof}
\begin{lem}\label{lem:logIneq}
        Let $y\geq x>0$. Then
  \begin{enumerate}[label=(\roman*)]
    \item
    \begin{equation}\label{eq:logineq2}
        \log{x+1\over x}>{2\over2x+1}>{1\over x+1},
    \end{equation}
    \item
    \begin{equation}\label{eq:logineq3}
        y\log{y+1\over y}\geq x\log{x+1\over x}.
    \end{equation}
  \end{enumerate}
\end{lem}
\begin{proof}
  (i) The rightmost inequality is a standard lower bound for the logarithm function. The first, tighter, inequality  is a special case of the logarithmic mean inequality~\cite{bullen1998dictionary}
  $$
  {z-x\over\log{z}-\log{x}}<{x+z\over2}
  $$
  valid for $z>x>0$, where we set $z=x+1$. 
  (ii) The function $f(x)=x\log{x+1\over x}$ is increasing for $x\in(0,\infty)$ since $f'(x)=\log{x+1\over x}-{1\over x+1}>0$ on this interval. This inequality follows from the second inequality in~(\ref{eq:logineq2}).
\end{proof}
\begin{lem}\label{lem:Npgreaterxm}
  Assume $1/2\leq\tau<1$ and $K\geq0$. Then $\eta\geq\ell$ holds whenever $N\geq0$.
\end{lem}
\begin{proof}
     In order to verify
     \begin{equation}\label{eq:NpMinusxm}
          \eta-\ell={1\over2}(1-N+3K+3N\tau-p)\geq0,
     \end{equation}
     where $p$ is given by Eq.~(\ref{eq:pakaD})
    \begin{equation}\label{eq:p}
         p=\big[(N+K+N\tau+1)^2-4\tau N(N+1)\big]^{1/2},
    \end{equation}
     we employ two convenient upper bounds on $p$. The first bound is $b_1=1+K+N+N\tau\geq p$ which follows from rewriting the above equation as
    \begin{equation}
      p = \big[b_1^2-4\tau N(N+1)\big]^{1/2}
    \end{equation}
     and the fact that $\tau,N,K\geq0$ and $p>0$. That itself follows (for $1/2\leq\tau\leq1$ and $K\geq0$) from yet another form of Eq.~(\ref{eq:p})   by inspection:
     \begin{equation}\label{eq:pAnother}
       p=\big[(1+K)^2+{2N} ({K} \tau + {K}- \tau +1)+{N}^2(1-\tau)^2\big]^{1/2}.
     \end{equation}
    The second bound is
    \begin{equation}\label{eq:b2upperbound}
         b_2=(1-\tau)N+{1+K-\tau+K\tau\over1-\tau}\geq p.
    \end{equation}
     To obtain the bound we first show that $p$ is concave as a function of $N>0$ for $K\geq0$ and $1/2\leq\tau<1$. It holds because
     $$
     \frac{\partial^2p}{\partial N^2}
     =-\frac{4 {K} \tau  ({K}-\tau +1)}{\big(({K}+{N} \tau +{N}+1)^2-4 \tau {N} ({N}+1)  \big)^{3/2}}=-{4 {K} \tau ({K}-\tau +1)\over p^3}<0.
     $$
     Then a tangent at any point $N>0$ is necessarily an upper bound. In particular, the tangent at $N\to\infty$ which happens to be the middle expression in Eq.~(\ref{eq:b2upperbound}) asymptotically approaches $p$ from above. So we get
     \begin{subequations}
       \begin{align}
          {1\over2}(1-N+3K+3N\tau-b_1) & =K+N\tau-N\geq0 \qquad\mbox{for}\quad0\leq N\leq{K\over1-\tau},\\
          {1\over2}(1-N+3K+3N\tau-b_2) & =(2\tau-1)\left({K\over\tau-1}+N\right)\geq0 \quad\mbox{for }N\geq{K\over1-\tau}
       \end{align}
     \end{subequations}
    for $1/2\leq\tau<1$.
\end{proof}
\begin{lem}\label{lem:uniformconv}
  Let $f(x)=\log{1+\a x\over\a x}, f_n(x)=\log{1+1/n+\a x\over 1/n+\a x}$ for $\a>0$ and  $\I^\e_x\df\{x\in\bbR;\e\leq x<\infty\}$ where $\e>0$. Then $f_n(x)\rightrightarrows f(x)$ on $\I^\e_x$ for all $\e>0$.
\end{lem}
\begin{proof}
    $\sup_{x\in\I^\e_x}|f_n(x)-f(x)|=0$ is an equivalent criterion for the uniform convergence of $f_n$ to $f$. This is indeed satisfied by inspection and by defining $g_n=f-f_n$  we explicitly find $g_n'=0\Leftrightarrow x_{extreme}=-1/(2\a)(1+1/n)$. Thus
  \begin{equation}\label{eq:uniformconv}
    \lim_{n\to\infty}\sup_{x\in\I^\e_x}|f_n(x)-f(x)|= \lim_{n\to\infty}\log{{(1+1/n-2\a)(-1+1/n(-1+2\a))\over(1+1/n)^2(-1+2\a)}}=0
  \end{equation}
  as long as $\a\neq1/2$.
\end{proof}
\begin{prop}\label{prop:GBehavior}
    Let $\GG(N,K,\tau)=g(\eta)-g(\ell)-g(f)$  be the non-optimized coherent information. Then $\GG(N,K,\tau)$  has at most one  stationary point for $N\in(0,\infty)$ whenever $1/2\leq\tau<1$ and $K\geq0$.
\end{prop}
\begin{proof}
  Stationary points are revealed by exploring $\papa{\GG}{N}=\papa{g(\eta)}{N}-\papa{g(\ell)}{N}-\papa{g(f)}{N}=0$ which is equivalent to
    \begin{equation}\label{eq:statCondOnN}
      \papa{\eta}{N}\log{1+\eta\over\eta}=\papa{f}{N}\log{1+f\over f}+\papa{\ell}{N}\log{1+\ell\over\ell}.
    \end{equation}
  Before we proceed to tackle the problem of how many times the above equality can be satisfied, we first take a look at the derivatives' behavior when $K=0$. In this case $f=0$, $\eta=N\tau$ and $\ell=N(1-\tau)$ and there is no stationary point since Eq.~(\ref{eq:statCondOnN}) becomes
  \begin{equation}\label{eq:StatInequality}
        \tau\log{1+N\tau\over N\tau}-(1-\tau)\log{1+N(1-\tau)\over N(1-\tau)}>0.
  \end{equation}
  The inequality holds for $1/2<\tau<1$ thanks to Lemma~\ref{lem:Npgreaterxm} and inequality~(\ref{eq:logineq3}). For $\tau=1/2$ we get $\GG(N,0,1/2)=0$.

  For $K>0$ the behavior of $\papa{g(\eta)}{N}$ radically changes for $N\gtrapprox 0$. By calculating the limits
    \begin{subequations}\label{eq:limits}
      \begin{align}
        \lim_{N\to0}\papa{g(\eta)}{N} & =\tau\log{1+K\over K},\label{eq:threeLimitsA} \\
        \lim_{N\to0}\papa{g(\ell)}{N} & = \infty, \label{eq:threeLimitsB}\\
        \lim_{N\to0}\papa{g(f)}{N} & = {K\tau\over1+K}\log{1+K\over K},\label{eq:threeLimitsC}
      \end{align}
    \end{subequations}
  we see that unlike for $K=0$, the expression $\lim_{N\to0}\papa{g(\eta)}{N}$  does not diverge. Hence
  \begin{equation}\label{eq:limitdifference}
    \lim_{N\to0}\papa{\GG(N,K,\tau)}{N}=-\infty
  \end{equation}
  whenever $K>0$. This is because of a ``jump'' from $+\infty$  when $K=0$ to a finite value for \emph{any} positive $K$  in Eq.~(\ref{eq:threeLimitsA}). But strictly speaking, the discontinuity occurs only for $N=0$ as can be seen from  $\papa{g(\eta)}{N}=\tau\log{1+K+N\tau\over K+N\tau}$ and this point is absent from $\I^\e_N=\I^\e_x$. By choosing $K$ arbitrarily small positive we must necessarily get $\papa{g(\eta)}{N}<\papa{g(\ell)}{N}+\papa{g(f)}{N}$ in an open neighborhood of $N=0$ (the non-negative complement of the set $\I^\e_N$). Therefore inequality Eq.~(\ref{eq:StatInequality}) obtained for $K=0$ gets ``immediately'' reversed but because $\GG$ is a continuous function of $K$ it is intuitively clear that for $K\gg0$ the curves will not be deformed much. More precisely, we may invoke Lemma~\ref{lem:uniformconv} showing uniform convergence of $\papa{g(\eta)}{N}$ on $\I^\e_N$ (set $x=N$, $n=1/K$ and $\a=\tau$). Since $\papa{g(\ell)}{N}+\papa{g(f)}{N}$  has no discontinuities for $N\geq0$ and $K\geq0$ then $\papa{\GG}{N}\big|_{K=0}$ is nicely approximated (i.e.~in the sense of the $\infty$-norm) by $\papa{\GG}{N}\big|_{K}$ on $\I^\e_N$ as $K\to0$. But this is not enough -- the uniform convergence does not inform us about a more detailed behavior. In particular, whether  the functions  intersect when approaching each other and how many times. The intuition suggests that the curves' order should be preserved for $K\gg0$ and so the number of intersections is odd and most likely just one.

  Let's put the intuition on a firm foundation. The main idea is based on an elementary calculus criterion for increasing/decreasing functions: If $w'(x)>0\ (w'(x)<0)$ for all $x\in(r,s)$, where the prime denotes differentiating, then $w(x)$ is increasing (decreasing) in $(r,s)$. Using this we will show that the LHS of Eq.~(\ref{eq:statCondOnN}) is a decreasing function of $K$  for all $N>0$ and $1/2<\tau<1$ while the RHS is increasing. Henceforth, they intersect just once. The LHS is easy to analyze:
  \begin{equation}\label{eq:LHSderwrtK}
    {\pa\over\pa K}\left[\papa{\eta}{N}\log{1+\eta\over\eta}\right]=-{\tau\over(K+N\tau)(1+K+N\tau)},
  \end{equation}
  which is negative for all investigated parameters. For the RHS we first observe
  \begin{equation}\label{eq:DerxmxpRelation}
    \papa{\ell}{N}=\papa{f}{N}+1-\tau.
  \end{equation}
  This allows us to write the RHS of Eq.~(\ref{eq:statCondOnN}) as
  \begin{align}
      \papa{f}{N}\left[\log{1+f\over f}+\log{1+\ell\over\ell}\right] +(1-\tau)\log{1+\ell\over\ell} & = h\log{1+F\over F}+(1-\tau)\log{1+\ell\over\ell},
  \end{align}
  where we defined
  \begin{subequations}\label{eq:Fh}
    \begin{align}\label{eq:F}
      F(N,K,\tau) & \df {f\ell\over f+\ell+1}={1+q-p\over2p},\\
      h(N,K,\tau) & \df \papa{f}{N}= {1\over2}\left(-1+\tau+{1+K+N+(-1+K-2N) \tau+N\tau^2\over p}\right).
    \end{align}
  \end{subequations}
  Eqs.~(\ref{eq:pakaD}) and (\ref{eq:q}) were used.   We differentiate w.r.t to $K$ and ask whether the following inequality holds:
  \begin{align}\label{eq:RHSderwrtK}
    {\pa\over\pa K}\bigg[h\log{1+F\over F}&+(1-\tau)\log{1+\ell\over\ell}\bigg] \nn \\
    &=\papa{h}{K}\log{1+F\over F}-h{\papa{F}{K}\over(1+F)F}-(1-\tau){\papa{\ell}{K}\over(1+\ell)\ell}\geq0.
  \end{align}
  To prove it we reshuffle things a bit by multiplying by $F$ and dividing by $\papa{h}{K}$. Both expressions are non-negative: $F$ is given by Eq.~(\ref{eq:F}) and we observe from the form of Eqs.~(\ref{eq:cohInfoParsB}) and (\ref{eq:cohInfoParsC})
    \begin{subequations}
    \begin{align}
      f & = {1\over2}(p-(1-\tau)N+K-1),\\
      \ell & = {1\over2}(p-(\tau-1)N-K-1).
    \end{align}
  \end{subequations}
  together with Eq.~(\ref{eq:pAnother}) that, because all the three summands in $p$ are always non-negative for $1/2\leq\tau\leq1$ and $K\geq0$, both $f$ and $\ell$ are manifestly non-negative too. From the first line of Eq.~(\ref{eq:q}) it also follows that $\papa{h}{K}$ given by Eq.~(\ref{eq:fourPapasC}) is non-negative.

  After the multiplications we use the first inequality (the tighter lower bound) in Eq.~(\ref{eq:logineq2}) to write:
  \begin{equation}\label{eq:RHSderwrtKbounded}
    F\log{1+F\over F}>{2F\over2F+1}\geq{h\over\papa{h}{K}}{\papa{F}{K}\over1+F}+(1-\tau){F\over\papa{h}{K}}{\papa{\ell}{K}\over(1+\ell)\ell}.
  \end{equation}
  Astonishingly, the second inequality is saturated. This is a non-trivial identity revealed only after a tedious calculation that we will not perform here. For the reader's convenience, we will list the remaining  expressions participating in~(\ref{eq:RHSderwrtKbounded}):
  \begin{subequations}\label{eq:fourPapas}
    \begin{align}
      {h\over\papa{h}{K}} & = {p^2\over2\tau(1+q)}\big(K(1+\tau)+(-1+\tau) (-1+N(-1+\tau)+p)\big),\\
      \papa{F}{K} & = {N(1+N)\over p^3}\big((-1+N(-1+\tau))(-1+\tau)+K(1+\tau)\big),\\
      \papa{h}{K} & ={\tau(q+1)\over p^3},\label{eq:fourPapasC}\\
      \papa{\ell}{K} & = {1+K+N+N\tau-p\over2p}.
    \end{align}
  \end{subequations}
  We conclude that the RHS of Eq.~(\ref{eq:statCondOnN}) is increasing as a function of $K$ for all $N$. Since $\papa{g(\eta)}{N}>\papa{g(\ell)}{N}$  holds for $K=0$ and $1/2<\tau<1$ (see Eq.~(\ref{eq:StatInequality})) and
  \begin{subequations}\label{eq:threeKinftyLims}
  \begin{align}
      \lim_{K\to\infty}\papa{g(\eta)}{N} & = 0, \label{eq:threeKinftyLimsA} \\
      \lim_{K\to\infty}\papa{g(\ell)}{N} & =\log{1+N\over N},  \label{eq:threeKinftyLimsB}\\
      \lim_{K\to\infty}\papa{g(f)}{N} & =0.  \label{eq:threeKinftyLimsC}
    \end{align}
 \end{subequations}
 we can see that both sides intersect exactly once.
\end{proof}
\begin{rem}
  We are not able to analytically solve Eq.~(\ref{eq:statCondOnN}) and find the intersection point (if we were able to do it we would likely not need the previous proposition). Just out of curiosity, numerical analysis suggests that as $N\to\infty$ the intersection point converges to a certain value as a function of $K$ and for a given $\tau$. This implies that there exists a threshold value $K_{th}>0$ where   if $K>K_{th}$ the function $\GG(N,K,\tau)$ is never positive. This is exemplified on the blue curve for a lossy OMG channel in Fig.~\ref{fig:Gexample}.
\end{rem}
\begin{rem}
   The proposition can also be applied for $0<\tau<1/2$ -- inequality Eq.~(\ref{eq:StatInequality}) gets reversed for $K=0$ leading to $\papa{g(\eta)}{N}<\papa{g(\ell)}{N}$  and the proved proposition informs us that the inequality will remain unchanged for $K>0$. For the illustration see the relevant part of the green region in Fig.~\ref{fig:cont} where the quantum capacity is known to be zero~\cite{caruso2006one} and so is the coherent information.
\end{rem}

Before we proceed with the proof of the main theorem we have to repeat the whole procedure for the class of amplifying quantum channels $\C(\mathrm{amp})$. The analysis is \emph{qualitatively}  similar  to the $\C(\mathrm{loss})$ case (including an intriguing identity \`{a} la the saturated second inequality in Eq.~(\ref{eq:RHSderwrtKbounded})) but it is different enough not to be omitted. We will be sketchy during some repetitive steps, though, such as the next lemma whose proof is nearly identical to the proof of Lemma~\ref{lem:limitsNinfty}. Recall that from this point onwards, our definition of $\eta$ becomes the second line of Eq.~(\ref{eq:cohInfoParsA}) and subsequently the remaining quantities in Eqs.~(\ref{eq:cohInfoPars}) and~(\ref{eq:pakaD}) change as well.
\begin{lem}\label{lem:limitsNinftyAmp}
  Assume $\tau>1$ and $K\geq0$. Then
  \begin{equation}\label{eq:NinftyGamp}
    \lim_{N\to\infty}\GG(N,K,\tau)= {K\over\tau-1}\log{K\over\tau-1}-{\tau-1+K\over\tau-1}\log{\tau-1+K\over\tau-1}
        +\log{\tau\over\tau-1}.
  \end{equation}
\end{lem}
Note that this result already appeared in~\cite{holevo2001evaluating}.
\begin{prop}\label{prop:GBehaviorAmp}
     The non-optimized coherent information $\GG(N,K,\tau)$  has at most one  stationary point for $N\in(0,\infty)$ whenever $\tau>1$ and $K\geq0$.
\end{prop}
\begin{proof}
  Similarly to Proposition~\ref{prop:GBehavior}, we will be asking  whether Eq.~(\ref{eq:statCondOnN}) can be satisfied and show that it happens at most at one point for $N\geq0$. We start with a simple case $K=0$ to reveal subtle differences compared to the lossy case. We can see that $\eta=(1+N)\tau-1,f=(1+N)(-1+\tau)$ and $\ell=0$. Therefore,
  \begin{align}\label{eq:StatIneqAmp}
    &\papa{\eta}{N}\log{1+\eta\over\eta}-\papa{f}{N}\log{1+f\over f}\nn\\
    &=\tau\log{(1+N)\tau\over(1+N)\tau-1}-(\tau-1)\log{(1+N)\tau-N\over(1+N)\tau-1-N}>0,
  \end{align}
  where the inequality follows from the logarithm properties and because of $\tau>1$. Contrary to the lossy case, when $K>0$ the behavior for $N\GtrSim0$ does not reverse the inequality since
  \begin{subequations}\label{eq:limitsAmp}
      \begin{align}
        \lim_{N\to0}\papa{g(\eta)}{N} & =\tau\log{K+\tau\over K-1+\tau},\label{eq:threeLimitsAAmp} \\
        \lim_{N\to0}\papa{g(f)}{N} & = {\tau(K-1+\tau)\over K+\tau}\log{K+\tau\over K-1+\tau} \label{eq:threeLimitsBAmp}
      \end{align}
  \end{subequations}
  and so clearly $\lim_{N\to0}\big[\papa{g(\eta)}{N}-\papa{g(f)}{N}\big]>0$ as in~(\ref{eq:StatIneqAmp}). However, setting $K>0$ makes $\ell$ and $g(\ell)$ nonzero for $N>0$. The function $g(\ell)$ is itself a well-behaved function for $N\geq0$ but its derivative has a discontinuity at $N=0$:
  $$
  \lim_{N\to0}\papa{g(\ell)}{N}  = 0\quad\mbox{for }K=0
  $$
  but
  \begin{equation}
    \lim_{N\to0}\papa{g(\ell)}{N}  = \infty\quad\mbox{for }K>0\label{eq:threeLimitsCAMP}.
  \end{equation}
  This is an equivalent of the discontinuity in the lossy case leading to $\papa{g(\eta)}{N}<\papa{g(\ell)}{N}+\papa{g(f)}{N}$ in a sufficiently small neighborhood of $N=0$. Hence, by a  different route compared to the lossy case, we again obtain
    \begin{equation}\label{eq:limitdifferenceAmp}
    \lim_{N\to0}\papa{\GG(N,K,\tau)}{N}=-\infty.
  \end{equation}

  The rest of the proof for the general case $K>0$ can be adapted almost verbatim from the lossy case starting with the equivalent of Eq.~(\ref{eq:LHSderwrtK}):
  \begin{equation}\label{eq:LHSderwrtKAmp}
    {\pa\over\pa K}\left[\papa{\eta}{N}\log{1+\eta\over\eta}\right]=-{\tau\over(-1+\tau+K+N\tau)(K+\tau+N\tau)},
  \end{equation}
  which is negative for all investigated parameters. Since the relation Eq.~(\ref{eq:DerxmxpRelation})  remains unchanged for the amplifying case, we can use the bound from Lemma~\ref{lem:logIneq} and again find the second inequality in Eq.~(\ref{eq:RHSderwrtKbounded}) saturated:
  \begin{equation}\label{eq:interestingIdentityAmp}
    {2F\over2F+1}={h\over\papa{h}{K}}{\papa{F}{K}\over1+F}+(1-\tau){F\over\papa{h}{K}}{\papa{\ell}{K}\over(1+\ell)\ell}.
  \end{equation}
  It contains the following expressions:
  \begin{subequations}\label{eq:fourPapasAmp}
    \begin{align}
      F(N,K,\tau) & = {q-p\over2p}, \\
      h(N,K,\tau) & = {1\over2}\left(-1+\tau+{{K(1+\tau)+(-1+\tau)(N(-1+\tau)+\tau)}\over p}\right),\\
      {h\over\papa{h}{K}} & = {p^2\over2\tau q}\big(K(1+\tau)+(-1+\tau)(N(-1+\tau)+\tau+p)\big),\\
      \papa{F}{K} & = {N(1+N)\over p^3}\big(K(1+\tau)+(-1+\tau)(N(-1+\tau)+\tau)\big),\\
      \papa{h}{K} & ={\tau q\over p^3},\\
      \papa{\ell}{K} & = {K+N+\tau+N\tau-p\over2p},
    \end{align}
  \end{subequations}
  where $q$ comes from the second line of Eq.~(\ref{eq:q}).
\end{proof}

\begin{proof}[Proof of the main theorem]
    Following Lemma~\ref{lem:limitsNinfty} and Lemma~\ref{lem:limitsNinftyAmp} we can write  (cf. \cite{holevo2001evaluating}):
    \begin{equation}\label{eq:GforNInfty}
        \lim_{N\to\infty}\GG(N,K,\tau)={K\over|1-\tau|}\log{K\over|1-\tau|}-{|1-\tau|+K\over|1-\tau|}\log{|1-\tau|+K\over|1-\tau|}
        +\log{\tau\over|1-\tau|}
    \end{equation}
    for both $1/2<\tau<1$ (the lossy class $\C(\mathrm{loss})$) and  $\tau>1$ (the amplifying class $\C(\mathrm{amp})$). The infinite limit can be both negative and positive as illustrated in Fig.~\ref{fig:Gexample} for the lossy case. If it is positive then a stationary point discovered in Proposition~\ref{prop:GBehavior} and~\ref{prop:GBehaviorAmp} for $K>0$ and $N\in(0,\infty)$ is necessarily a global minimum whose value is negative (see the grey curve in Fig.~\ref{fig:Gexample} as an example) as implied by Eq.~(\ref{eq:limitdifference}) in the lossy case and Eq.~(\ref{eq:limitdifferenceAmp}) in the amplifying case.
    If $\lim_{N\to\infty}\GG(N,K,\tau)<0$, the coherent information cannot be positive for any $N\in(0,\infty)$ since again there is at most one stationary point in this interval and it is always negative. This possibility is illustrated in Fig.~\ref{fig:Gexample} by the green and blue curves. Then, from Eqs. (\ref{eq:Icoh}),~(\ref{eq:g}) and~(\ref{eq:cohInfoPars}), and by using the standard limit $\lim_{x\to0}[x\log{x}]=0$, we find
    \begin{equation}\label{eq:GforNZero}
      \lim_{N\to0}\GG(N,K,\tau)=0.
    \end{equation}
    For $K=0$ the situation is even simpler. There is no stationary point for $N\in(0,\infty)$ and the coherent information is monotone decreasing on this interval with a (zero) supremum at $N\to0$. This concludes the proof.
\end{proof}

\begin{figure}[t]
  \resizebox{9cm}{!}{\includegraphics{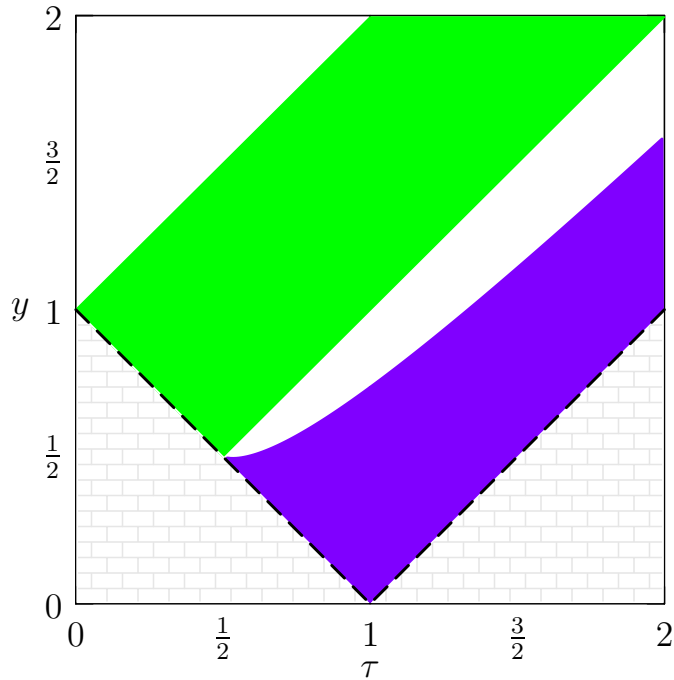}}
   \caption{A section of the parameter region covering three OMG channel classes is depicted: $\C(\mathrm{loss})$ for $0<\tau<1$, $\B_2$ for $\tau=1$ and $y>0$ and $\C(\mathrm{amp})$ for $\tau>1$. The purple area contains OMG channels whose coherent information (and therefore quantum capacity) is positive. The green area are \emph{antidegradable} channels~\cite{devetakshor2005capacity,caruso2006one}, whose quantum capacity (and therefore coherent information) is zero. The channels from the top-left white triangle form a subclass of zero quantum capacity channels known as entanglement-breaking channels~\cite{holevo2007one} and this region extends further in the parameter space. The channels from the white region in between the green  and purple areas have zero coherent information but their quantum capacity is unknown. All the three regions continue indefinitely for $\tau>0$. The dashed lines are given by the condition $K=0$ for the added classical noise. The brickwall (also extending indefinitely) denotes an unphysical region where Gaussian maps are not completely positive. }
    \label{fig:cont}
\end{figure}

\begin{rem}\label{rem:downwards}
    If the limit Eq.~(\ref{eq:GforNInfty})  is negative, the following three possibilities could have in principle occurred: (i) there is still a global minimum but the non-optimized coherent information $\GG$ does not cross zero and remains negative (green curve in Fig.~\ref{fig:Gexample}), (ii) there is no stationary point for $N\in(0,\infty)$ and $\GG$~is negative and monotone decreasing (blue curve) and, (iii) there is an inflection point and $\GG$ again remains negative. Numerical analysis suggests that option (iii) does not happen.
\end{rem}
Even though the case $\tau=1$ belongs to a formally different class of channels $\B_2$, called the additive classical noise channel class, nothing dramatic happens in the above analysis as long as $K\neq0$. As a matter of fact, the expressions in Proposition~\ref{prop:GBehaviorAmp} simplify and manifestly keeps on holding as we approach $\tau=1$ from the right. Hence we have the following Corollary.
\begin{cor}
    \begin{equation}\label{eq:cortau1}
        I_{\rm coh}(\B_2)=\sup_{N}\GG(N,K,1)=\lim_{N\to\infty}\GG(N,K,1)=-1-\log{K}.
    \end{equation}
    holds for all $K>0$.
\end{cor}
The identity channel is obtained for $K=0$ where $I_{\rm coh}(\B_2)$ diverges~\cite{holevo2001evaluating}. Hence the performed analysis also covered this special case of $\B_2$ where $\ave{n}=0$ (and so $r=1$).

\section{Discussion}

In this work we proved that the coherent information of phase-insensitive one-mode Gaussian (OMG) channels  is maximized in the limit of an infinite input signal power. This has been  known to trivially hold in the case of zero added classical noise $(K=0)$ but the nonzero added noise case is less obvious thanks to a non-trivial dependence of the coherent information $\GG(N,K,\tau)$ on the input signal power $N$. The proof was made relatively straightforward by virtue of an interesting identity involving quantities derived from the von Neumann entropy for bosonic Gaussian systems. The identity may eventually become useful elsewhere.

In a sense, the current analysis covers \emph{almost all} phase-insensitive OMG channels (using the measure-theoretic terminology) whose coherent information is non-vanishing. The channels with zero added classical noise form a mere boundary of the region delimiting  the channels with nonzero added  noise. This can be best visualized by a reparametrization provided by the parameters $\tau\in\bbR$ and $y\geq0$ defined in Section~\ref{sec:main}, that first appeared in~\cite{schafer2013} and was also used in~\cite{schafer2013equivalence,giovannetti2013ultimate}. By performing the optimization and finding that the coherent information is maximized for $\lim_{N\to\infty}\GG(N,K,\tau)$, we can investigate when the limit is non-negative and  obtain a plot identical to a figure from~\cite{bradler2014black} where, incidentally, the currently studied OMG channels appeared in the context of black holes physics and the related issue of information loss. Needless to say that the study presented here is independent of any such physical interpretation.

Looking at Fig.~\ref{fig:cont} we can appreciate two facts. As previously mentioned, we realize the scope of the main theorem proved in this paper. The channels with zero added noise ($K=0$), whose coherent information is easy to optimize~\cite{holevo2012quantum}, are only those lying on the dashed boundary. The OMG channels treated in this paper form an infinite quarter-plane given by $y>|\tau-1|$. However, only the region where $\tau>0$ is interesting since the class of conjugate amplifying OMG channels $\D$, defined for negative $\tau$, is known to be entanglement breaking and therefore cannot be used for reliable quantum communication.

For $\tau>1/2$, several different possibilities occur. The purple region denotes OMG channels whose optimized coherent information is positive  (and so is the quantum capacity). The green region contains so-called antidegradable channels~\cite{devetakshor2005capacity,caruso2006one} and their quantum capacity (and therefore the coherent information as well) is zero. The white wedge-like region in between is a zero coherent information region but the quantum capacity is unknown. In fact, the result of the main theorem presented here can be summarized by saying that in the purple area  the  coherent information $\GG(N,K,\tau)$ cannot be made greater by a different choice of the input power $N$ other than $N\to\infty$ and, conversely, the coherent information cannot be made nonzero in the white area  by any choice of $N>0$.

For $0<\tau\leq1/2$ the OMG channels are either antidegradable (the green region) or entanglement breaking (the corresponding part of the top-left white triangle) and so their coherent information is zero and the quantum capacity as well.

\section*{Acknowledgement}
\thanks{The author thanks an anonymous referee for comments.}

\bibliographystyle{unsrt}


\begin{thebibliography}{10}

\bibitem{devetak2005capacity}
Igor Devetak.
\newblock The private classical capacity and quantum capacity of a quantum
  channel.
\newblock {\em IEEE Transactions on Information Theory}, 51:44--55, 2005.

\bibitem{shor2002quantum}
Peter~W Shor.
\newblock The quantum channel capacity and coherent information.
\newblock In {\em Lecture notes, MSRI Workshop on Quantum Computation}, 2002.

\bibitem{lloyd1997capacity}
S~Lloyd.
\newblock Capacity of the noisy quantum channel.
\newblock {\em Physical Review A}, 55(3):1613, 1997.

\bibitem{barnum1998information}
Howard Barnum, Michael~A Nielsen, and Benjamin Schumacher.
\newblock Information transmission through a noisy quantum channel.
\newblock {\em Physical Review A}, 57(6):4153, 1998.

\bibitem{holevo2012quantum}
Alexander~S Holevo.
\newblock {\em Quantum Systems, Channels, Information: A Mathematical
  Introduction}.
\newblock De Gruyter, 2012.

\bibitem{divincenzo1998quantum}
David~P DiVincenzo, Peter~W Shor, and John~A Smolin.
\newblock Quantum-channel capacity of very noisy channels.
\newblock {\em Physical Review A}, 57(2):830, 1998.

\bibitem{smith2011quantum}
Graeme Smith, John~A Smolin, and Jon Yard.
\newblock {Quantum communication with Gaussian channels of zero quantum
  capacity}.
\newblock {\em Nature Photonics}, 5(10):624--627, 2011.

\bibitem{weedbrook2012gaussian}
Christian Weedbrook, Stefano Pirandola, Ra\'ul Garc\'ia-Patr\'on, Nicolas~J
  Cerf, Timothy~C Ralph, Jeffrey~H Shapiro, and Seth Lloyd.
\newblock Gaussian quantum information.
\newblock {\em Reviews of Modern Physics}, 84(2):621, 2012.

\bibitem{holevo2007one}
Alexander~S Holevo.
\newblock {One-mode quantum Gaussian channels: Structure and quantum capacity}.
\newblock {\em Problems of Information Transmission}, 43(1):1--11, 2007.

\bibitem{giovannetti2013ultimate}
V~Giovannetti, R~Garc\'ia-Patr\'on, NJ~Cerf, and Alexander~S Holevo.
\newblock {Ultimate communication capacity of quantum optical channels by
  solving the Gaussian minimum-entropy conjecture}.
\newblock {\em Nature Photonics}, 8:796, 2014.

\bibitem{schafer2013equivalence}
Joachim Sch{\"a}fer, Evgueni Karpov, Ra{\'u}l Garc{\'\i}a-Patr{\'o}n, Oleg~V
  Pilyavets, and Nicolas~J Cerf.
\newblock {Equivalence relations for the classical capacity of single-mode
  Gaussian quantum channels}.
\newblock {\em Physical Review Letters}, 111(3):030503, 2013.

\bibitem{pilyavets2012methods}
Oleg~V Pilyavets, Cosmo Lupo, and Stefano Mancini.
\newblock {Methods for estimating capacities and rates of Gaussian quantum
  channels}.
\newblock {\em IEEE Transactions on Information Theory}, 58(9):6126--6164,
  2012.

\bibitem{gioCMP}
V.~Giovannetti, A.S. Holevo, and R.~Garc\'ia-Patr\'on.
\newblock A solution of gaussian optimizer conjecture for quantum channels.
\newblock {\em Communications in Mathematical Physics}, pages 1--19, 2014.

\bibitem{holevo2001evaluating}
Alexander~S Holevo and Reinhard~F Werner.
\newblock {Evaluating capacities of bosonic Gaussian channels}.
\newblock {\em Physical Review A}, 63(3):032312, 2001.

\bibitem{cerf2007quantum}
J.~Eisert and M.~M. Wolf.
\newblock Gaussian quantum channels.
\newblock In N.~J. Cerf, G.~Leuchs, and E.~Z. Polzik, editors, {\em Quantum
  Information with Continuous Variables of Atoms and Light}, pages 23--42.
  World Scientific, 2007.

\bibitem{wolf2007quantum}
Michael~M Wolf, David P{\'e}rez-Garc{\'\i}a, and Geza Giedke.
\newblock Quantum capacities of bosonic channels.
\newblock {\em Physical Review Letters}, 98(13):130501, 2007.

\bibitem{caruso2006one}
Filippo Caruso, Vittorio Giovannetti, and Alexander~S Holevo.
\newblock {One-mode bosonic Gaussian channels: A full weak-degradability
  classification}.
\newblock {\em New Journal of Physics}, 8(12):310, 2006.

\bibitem{devetakshor2005capacity}
Igor Devetak and Peter~W Shor.
\newblock The capacity of a quantum channel for simultaneous transmission of
  classical and quantum information.
\newblock {\em Communications in Mathematical Physics}, 256(2):287--303, 2005.

\bibitem{caruso2006degradability}
Filippo Caruso and Vittorio Giovannetti.
\newblock {Degradability of bosonic Gaussian channels}.
\newblock {\em Physical Review A}, 74(6):062307, 2006.

\bibitem{bradler2014black}
Kamil Br\'adler and Christoph Adami.
\newblock {Black holes as bosonic Gaussian channels}.
\newblock {\em arXiv preprint arXiv:1405.1097}, 2014.

\bibitem{schafer2013}
J~Sch\"afer.
\newblock {\em {Information transmission through bosonic Gaussian quantum
  channels}}.
\newblock PhD thesis, {Universit\'e libre de Bruxelles}, 2013.

\bibitem{agarwal}
G.~S. Agarwal.
\newblock Entropy, the {W}igner distribution function, and the approach to
  equilibrium of a system of coupled harmonic oscillators.
\newblock {\em Physical Review A}, 3:828--831, 1971.

\bibitem{bullen1998dictionary}
Peter Bullen.
\newblock {\em A dictionary of inequalities}.
\newblock CRC Press, 1998.

\end{thebibliography}

\end{document}